\documentclass[a4paper,USenglish,cleveref,thm-restate]{lipics-v2021}
\pdfoutput=1 

\usepackage{style}
\usepackage{shortcuts}

\title{Cost Preserving Dependent Rounding for Allocation Problems}

\author{Lars Rohwedder}{University of Southern Denmark, Odense, Denmark}{rohwedder@sdu.dk}{}{}
\author{Arman Rouhani}{Maastricht University, The Netherlands}{a.rouhani@maastrichtuniversity.nl}{}{}
\author{Leo Wennmann}{University of Southern Denmark, Odense, Denmark}{wennmann@imada.sdu.dk}{}{}
\authorrunning{L.~Rohwedder, A.~Rouhani and L.~Wennmann} 
\Copyright{Lars Rohwedder and Arman Rouhani and Leo Wennmann} 

\funding{\textit{Lars Rohwedder and Leo Wennmann}: Supported by Dutch Research Council (NWO) project "The Twilight Zone of Efficiency: Optimality of Quasi-Polynomial Time Algorithms" [grant number OCEN.W.21.268].}

\ccsdesc[500]{Theory of computation~Rounding techniques}
\keywords{Matching, Randomized Rounding, Santa Claus, Approximation Algorithms}

\category{Track A: Algorithms, Complexity and Games}
\relatedversion{}
\relatedversiondetails{Full Version}{https://arxiv.org/abs/2502.08267}

\EventEditors{Keren Censor-Hillel, Fabrizio Grandoni, Joel Ouaknine, and Gabriele Puppis}
\EventNoEds{4}
\EventLongTitle{52nd International Colloquium on Automata, Languages, and Programming (ICALP 2025)}
\EventShortTitle{ICALP 2025}
\EventAcronym{ICALP}
\EventYear{2025}
\EventDate{July 8--11, 2025}
\EventLocation{Aarhus, Denmark}
\EventLogo{}
\SeriesVolume{334}
\ArticleNo{122}

\nolinenumbers %uncomment to disable line numbering
\begin{document}

\maketitle

\begin{abstract}
    We present a dependent randomized rounding scheme, which rounds fractional solutions to integral solutions satisfying certain hard constraints on the output while preserving Chernoff-like concentration properties.
    In contrast to previous dependent rounding schemes, our algorithm guarantees that the cost of the rounded integral solution does not exceed that of the fractional solution.
    Our algorithm works for a class of assignment problems with restrictions similar to those of prior works.

    In a non-trivial combination of our general result with a classical approach from Shmoys and Tardos [Math.\;Programm.'93] and more recent linear programming techniques developed for the restricted assignment variant by Bansal, Sviridenko [STOC'06] and Davies, Rothvoss, Zhang~[SODA'20], we derive a $O(\log n)$-approximation algorithm for the \emph{Budgeted Santa Claus Problem}.
    In this new variant, the goal is to allocate resources with different values to players, maximizing the minimum value a player receives, and satisfying a budget constraint on player-resource allocation costs.
\end{abstract}
\section{Introduction}
\label{sec:introduction}
A successful paradigm in the design of approximation algorithms
is to first solve a continuous relaxation, which can
typically be done efficiently using
linear programming, and then to round the 
fractional solution $x\in [0,1]^d$
to an integer solution $X\in\{0,1\}^d$. 
Careful choices need to be made in the rounding step
so that the error introduced is low.
Independent randomized rounding is one of the most natural
rounding schemes. In the simplest variant, we independently
set each variable $X_i$ to $1$ with probability $x_i$ and
to $0$ with probability~$1 - x_i$.
The advantage is that the value of every linear function
(over the $d$ variables) is maintained in expectation.
Moreover, for linear functions with small coefficients,
a Chernoff bound yields strong concentration guarantees
for the value.
Hence, if the initial solution $x$ satisfies some linear 
constraints from the continuous relaxation, we can often
argue with several Chernoff bounds combined by a union bound that 
they are still satisfied up to a small error by the
rounded solution $X$.

In some cases, however, the problem dictates structures or
hard constraints on the solution. For example, we might
require $X$ to be (the incidence vector of) a perfect matching
in a given graph or the basis of a given matroid.
Perfect matchings or bases are objects
that are quite simple in many computational aspects,
but it is typically very unlikely that independent
randomized rounding on a fractional object,
that is, a point in the convex hull of the objects we
want, results in one of these objects.
This motivates so-called dependent randomized rounding.
Here, the goal is to achieve similar guarantees
as independent randomized rounding, but with a distribution over a restricted set of objects,
which necessarily introduces some dependency between
variables.
These rounding schemes are typically tailored to specific
object structures and achieving comparable goals is already challenging for very simple structures.

For bipartite perfect matchings, a fundamental structure in
combinatorial optimization, one cannot hope to achieve
similar concentration guarantees to independent randomized rounding, due to the following well-known example.
Given a cycle of length $n\in 2\N$,
there are only two perfect matchings, 
the two alternating sets of edges. If the values
of a linear function over the edges alternate between
$0$ and $1$, then the fractional matching, which takes every
edge with $1/2$ will have a function value of $n/2$, but
each of the integral matchings incurs an additive
distortion of $\Omega(n)$, much higher than the bound
of $O(\sqrt{n})$ that holds with high probability
if each edge is picked independently with probability $1/2$.
If one considers~$b$-matchings or other more general assignment problems, however, then there are non-trivial guarantees that can be achieved with dependent rounding.
Gandhi, Khuller, Parthasarathy, and Srinivasan~\cite{GandhiKPS06} 
show that between any two edges incident to the same vertex, they can establish \emph{negative correlation}. 
Furthermore, their algorithm has the natural property of \emph{marginal preservation},
which means that the probability of $X_e = 1$ is equal to the fractional value $x_e$ for each variable $X_e$.
Together this implies strong concentration guarantees at least for linear functions on the incident edges of each vertex.
The following proposition is a consequence of their result.
\begin{proposition}\label{prop:assign}
    Let $G = (A\cup B, E)$ be a bipartite graph and
    $x\in [0, 1]^{E}$ represent a fractional many-to-many assignment.
    Furthermore, let $c\in\RR^{E}$, and $a_1,\dotsc,a_k\in [0, 1]^E$ such
    that for each $j\in \set{1,\dotsc,k}$ there is some $v\in A\cup B$
    with $\supp(a_j) \subseteq \delta(v)$.
    Then, in randomized polynomial time, one can compute
    $X\in \{0, 1\}^E$ satisfying with constant probability
    \begin{description}
        \item[Cost Approximation.] \hspace{0.13cm} $c\T X \le (1 + \epsilon) \cdot c\T x$
        \item[Concentration.] \hspace{1.03cm} $|a_j\T x - a_j\T X| \le
            O(\max\{\log k, \sqrt{a_j\T x \cdot \log k}\})$ for all $j\in \set{1,\dotsc,k}$
        \item[Degree Preservation.] \hspace{0.05cm}  $X(\delta(v)) \in \{\floor{x(\delta(v)}, \ceil{x(\delta(v))} \}$
    \end{description}
    Here, $\epsilon > 0$ is an arbitrarily small constant
    that influences the success probability.
\end{proposition}

A generalization to matroid intersection with a similar
restriction was shown by Chekuri, Vondrack, and Zenklusen~\cite{ChekuriVZ10}.
The same work also presents a dependent rounding
scheme for a single matroid that outputs a basis
satisfying similar concentration bounds on linear functions without a restriction on the support.
In this study, we ask the following question:
\begin{center}
    \medskip
    \emph{Can we avoid an error in the cost for dependent rounding while maintaining comparable other guarantees?}
    \medskip
\end{center}
We call a rounding algorithm \emph{cost preserving} if it does not
exceed the cost of the fractional solution we start with.
Here, we focus on the stronger variant where distributions are only over objects that are cost preserving, although one might be satisfied with a sufficiently high probability of cost preservation in some
cases. We have no evidence that such a relaxation would make the task significantly easier.

There are several situations where even the seemingly
small cost approximation of $(1 + \epsilon)$, as derived from
marginal preservation and Markov's inequality in the previously mentioned result, is unacceptable. For example, the cost of the fractional solution might come from a hard budget constraint $c\T x \le C$ in the problem.
Another situation is an extension of the objective function to
potentially negative values, representing for example the task of
maximizing profit = revenue - cost. Here, Markov's inequality
cannot be applied at all. Finally, an algorithm that preserves the cost
provides polyhedral insights: every fractional object
is in the convex hull of integer objects that marginally deviate
 in the considered linear functions. And similarly, the (non-integral) polytope of a relaxation is
contained in an approximate integral polytope. 
It is easy to see that cost preservation is incompatible with marginal preservation and hence cannot be satisfied by the dependent rounding schemes above: consider $d+1$ variables~$x_0, x_1, \dotsc, x_d$
of which exactly one is selected, then this can be modeled by bases of a uniform matroid or a degree constraint in the assignment problem above.
Suppose that $c_1 = \cdots = c_{d} = 1/(1 - 2^{-d}) > 1$ and $c_0 = 0$ where the fractional solution is given by~$x_1 = \cdots = x_n = (1 - 2^{-d})/d$
and $x_0 = 2^{-d}$, leading to a cost of~$1$. For a marginal preserving distribution, the probability that the integral solution $X$ has a cost lower than~$1$ (i.e.,~$X_0 = 1$) is exponentially small.
Note, however, that this is not an immediate counter-example to our stated goal:
in this example, deterministically choosing $X_0 = 1$ (and~$X_1 = \cdots = X_d = 0$) still maintains
$|a\T x - a\T X| \le 1$ for every $a\in [0, 1]^{d+1}$. 
\paragraph*{Our contributions}
Our results are twofold.
First, we show that one can obtain comparable guarantees to \Cref{prop:assign} while preserving costs. 
\begin{restatable}{theorem}{RDR}
    \label{thm:assign}
    Let $G = (A\cup B, E)$ be a bipartite graph and
    $x\in [0, 1]^{E}$ represent a many-to-many assignment.
    Furthermore, let $c\in\R^{E}$ and $a_1,\dotsc,a_k\in [0, 1]^E$ such
    that for each $j \in \set{1,\dotsc,k}$ there is some $v\in A\cup B$
    with $\supp(a_j) \subseteq \delta(v)$.
    Then, in randomized polynomial time, one can compute
    $X\in \{0, 1\}^E$ satisfying with constant probability
    \begin{description}
        \item[Cost Preservation.] \hspace{0.46cm} $c\T X \le c\T x$,
        \item[Concentration.] \hspace{1.03cm} $|a_j\T x - a_j\T X| \le
            O(\max\{\log k, \sqrt{a_j\T x \cdot \log k}\})$ for all $j\in \set{1,\dotsc,k}$,
        \item[Degree Preservation.] \hspace{0.05cm} $\sss{X}{\delta(v)} \in \{\floor{\, \sss{x}{\delta(v)} \,}, \ceil{\, \sss{x}{\delta(v)} \,} \}$.
    \end{description}
\end{restatable}
\medskip \noindent Note that in contrast to the previous result, we allow for negative components in the cost function~$c$.

Second, we present a non-trivial application of our theorem to an allocation problem we call the \emph{Budgeted Santa Claus Problem (with identical valuations)}.
Colloquially, it is often described as Santa Claus distributing gifts to children on Christmas.
Formally, there are $n$ resources~$\cR$ (gifts) to be distributed among $m$ players~$\cP$ (children). 
Each resource~$j$ has a specific value~$v_j \ge 0$.
Additionally, there is a total budget of $C \ge 0$, and assigning a resource~$j$ to a player $i$ incurs
a cost denoted by $c_{ij} \ge 0$. 
The goal is a distribution of resources among the players where the least happy player is as happy as possible and
the total cost does not exceed the budget $C$. 
Formally, we aim to find disjoint sets $R_i\subseteq \cR $, $i\in \cP$, maximizing $\min_{i\in\cP} \sum_{j \in R_i} v_{j}$
% \begin{equation*}
%     \min_{i\in\cP} \sum_{j \in R_i} v_{j}
% \end{equation*}
while ensuring that $\sum_{i\in\cP }^{}\sum_{j\in R_i}^{} c_{ij} \leq C$.
Note that not all resources need to be assigned. However, the variant, where all resources must be assigned can be shown to be 
not more difficult than our problem, see Appendix~\ref{subsec:appendix-santa-claus-with-all-resources-assigned}.

It is possible to consider an even more general variant
where each value $v_{ij}$ depends on both player $i$~and
resource~$j$, which we call the \emph{unrelated valuations}.
Mainly, we restrict ourselves to identical valuations because
the understanding of unrelated valuations in literature is
rather poor---even without considering costs.
In fact, much of the recent literature is focused on the so-called
restricted assignment case
of unrelated valuations (without costs),
where~$v_{ij}\in\{0, v_j\}$,
meaning each resource is either desired with a value of~$v_j$ or worthless to a player.
Among players who desire a particular resource,
its value is the same. Our budgeted variant generalizes the
restricted assignment case: observe that by setting costs~$c_{ij}\in\{0,1\}$ and $C = 0$, we can restrict
the set of players to which a resource can be assigned.
In a non-trivial framework, we apply our dependent rounding theorem
to obtain the following approximation guarantee.
\begin{restatable}{theorem}{SCapprox}
    There is a randomized polynomial time $O(\log n)$-approximation algorithm for the Budgeted Santa Claus problem.
\end{restatable}

\paragraph*{Other related work for dependent rounding}
Saha and Srinivasan~\cite{SahaS18} also provide a
dependent rounding scheme for allocation problems,
focusing on combinations of dependent
and iterative rounding.
Bansal and Nagarajan~\cite{BansalN16}
combine dependent rounding with techniques from discrepancy theory, known as the Lovett-Meka algorithm~\cite{LovettM15}. They
prove that one can round a fractional independent
set (or basis) of a matroid to an integral one,
while maintaining comparable concentration guarantees to both Lovett-Meka and Chernoff-type bounds.
We note that Bansal and Nagarajan also integrate
costs in their framework, but they make the
assumption that the costs are polynomially bounded,
which is inherently different from our setting (apart from the fact that they consider matroids).

Another well-known dependent rounding scheme is the maximum entropy
rounding developed by Asadpour, Goemans, M{\k{a}}dry, Gharan and Saberi~\cite{AsadpourGMGS17}.
This is used to sample
a spanning tree, i.e., a basis of a particular matroid,
while guaranteeing negative correlation properties and therefore Chernoff-type concentration.
This result led to the first improvement over the longstanding
approximation rate of $\Theta(\log n)$ for the asymmetric
traveling salesman problem (ATSP).
However, all algorithms above guarantee marginal preservation, which means they
cannot guarantee cost preservation.

At least superficially related to our work is the literature on multi-budgeted
independence systems~\cite{ChekuriVZ11, GrandoniZ10}.
Here, the goal is to find objects of certain structures, e.g.,
matchings or independent sets of matroids,
subject to several (potentially hard) packing constraints
of the form
$a\T x\le b$ for some $a\in \RR^n$, $b\in \RR$.
This can also be used to model cost preservation alike to our
results.
Chekuri, Vondrak, and Zenklusen~\cite{ChekuriVZ11}
and Gradoni and Zenklusen~\cite{GrandoniZ10}
show various positive results in a similar spirit to ours.
These results, however, are restricted to downward-closed structures where for a given solution, formed by a
set of elements, all subsets are valid solutions as well.
For example, Chekuri, Vondrak, and Zenklusen achieve
strong concentration results for randomized rounding on matchings, but this relies on dropping edges in long augmenting paths or cycles in order to reduce dependencies.
Gradoni and Zenklusen~\cite{GrandoniZ10} give a rounding algorithm for a constant number of hard budgets, but
this requires rounding down all components of a fractional solution.
Hence, these results are unable to handle instances like matroid basis constraints, perfect matching constraints, or degree preservation as in~\Cref{thm:assign}.

\paragraph*{Other related work for the Santa Claus problem}
Omitting the costs in the variant we study, the problem
becomes significantly easier and admits an EPTAS, see e.g.~\cite{JansenKV20}, which relies on techniques that contrast with the ones that are relevant to us.
As mentioned before, the problem with costs generalizes
the restricted assignment variant and therefore
inherits the approximation hardness of $2-\epsilon$ due to~\cite{BansalS06}.
Here and in the following, we use restricted assignment synonymous
with the variant without costs, but $v_{ij}\in \{0, v_j\}$.
Bansal and Srividenko~\cite{BansalS06} developed a randomized rounding algorithm for the restricted assignment.
Normally, this would lead
to a similar logarithmic approximation rate as ours
(for the problem without costs),
but they show that combining it
with the Lov\'asz Local Lemma yields an even better rate of
$O(\log\log m/ \log\log\log m)$.
Using similar techniques, the rate was improved to a constant by Feige~\cite{Feige08}.
Note that this randomized rounding uses intricate preprocessing
that violates the marginal preserving property and thus cannot even
maintain the cost of a solution in expectation.
Based on local search, there is also a combinatorial approach, 
see e.g.,~\cite{BamasLMRS24, AsadpourFS08, AnnamalaiKS17},
which yields a (better) constant approximation for restricted
assignment. However, it is not at all clear how costs could be
integrated in this framework.

Finally, a classical algorithm by Shmoys and Tardos~\cite{ShmoysT93}
gives an additive guarantee, where the rounded integral solution
is only worse by the maximum value $v_{\max} = \max_{ij} v_{ij}$.
Therefore, it even works in the unrelated case without increasing
the cost. Notably, they state this result for
the dual of minimizing the maximum value, namely the Generalized Assignment Problem.
The mentioned guarantee for Santa Claus is followed by a trivial
adaption, see \Cref{lem:rounding-small-items}.
Although very influential, this is the only technique we are
aware of which considers the problem with costs.
Unfortunately, this additive guarantee does not lead
to a multiplicative guarantee, since the optimum may be lower
than~$v_{\max}$. In fact, it is well known that the linear 
programming relaxation used in~\cite{ShmoysT93} has an unbounded
integrality gap even for restricted assignment~\cite{BansalS06}. Hence, one
cannot hope to improve this by a simple modification.
Nevertheless, this algorithm forms an important subprocedure in our result.

\paragraph*{Notation}

First, we introduce some necessary notation. Let $S,T \in \set{0,1}^E$ be edge sets in a bipartite graph $G = (A\cup B, E)$.
For all $T \subseteq S$, define $S(T) = \sum_{e \in T} S_e$. 
Let $P$ be the convex hull of degree preserving edge sets $S \in [0,1]^E$. Moreover, for any $v \in A\cup B$ define $\delta(v) = \{e \in E \mid v$ is incident to $e\}$. For the sake of simplicity, we use the shorthand notation $[q] = \set{1, \dotsc, q}$ for any $q \in \N$. Furthermore, for any vector $a \in [0, 1]^E$, the support of $a$ is denoted by $\supp(a) = \{e \in E \mid a_e \neq 0 \}$.
\section{Budgeted Dependent Rounding}
\label{sec:dependent-randomized-rounding}

This section will introduce a dependent randomized rounding procedure,
which produces an integral solution satisfying certain concentration guarantees, while preserving the cost and the degree of the fractional solution.
The formal properties are summarized in the following theorem.

\RDR*
\medskip

Throughout this section, the proofs of the technical lemmas are deferred to~\cref{subsec:appendix-ommitted-proofs-of-dependent-rounding}.
An oversimplified outline of our algorithm is as follows:
imagine~$x$ is
the average of two integral edge sets, then
the result can be shown by decomposing the
symmetric difference of both edge sets into cycles and paths.
We reduce to this case by starting with many edge sets and iteratively merging pairs of them in
a tree-like manner.

In order to find the initial integral edge sets, we compute a representation of~$x$ (or rather another similar assignment~$y'$) that is a convex
combination of degree preserving edge sets
such that its scalars satisfy
a certain level of discreteness.
Let~$P$ be the convex hull of degree preserving edge sets~$S \in \set{0,1}^E$, that is, those $S$ that satisfy for all $v\in A\cup B$
\begin{equation*}
    S(\delta(v))\in\{\floor{x(\delta(v))}, \ceil{x(\delta(v)}\}.
\end{equation*}
It can be shown that $x$ is contained in $P$ and,
in particular, $x$ is a convex combination of
degree preserving sets. 
In the following lemma, we show something even stronger: there
exists a fractional assignment~$y$ at least
as good as~$x$, which is the
convex combination of only few edge sets
and has few fractional
variables in the support of each constraint.

\begin{restatable}{lemma}{BoundedNumberOfFracVariables}
    \label{lem:bounded-number-of-frac-variables}
    There exists a convex combination $y = \sum_{i \in [k]} \lambda_i S_i$
    where $\lambda_i \in [0,1]$ and $\sum_{i \in [k]} \lambda_i = 1$ and $S_i \in P \cap \{0,1\}^E$ with
    \begin{align}
        c\T y &\le c\T x & \\
        a_j\T y &= a_j\T x &&\forall j \in \set{1,\dotsc,k} \\
        |\{e\in \delta(v) \mid y_e \notin \{0, 1\}\}| &\le 2k &&\forall v\in A\cup B & \label{eq:sos-bounded-support}
    \end{align}
\end{restatable}
Considering~$y$ as a vertex solution of a linear program, the proof follows from analyzing the structure of polytope~$P$.
For our algorithm, however, the scalars~$\lambda_i$ are not discrete enough.
Hence, we use the following lemma to round~$y$ to a more discrete assignment~$y'$.

\begin{restatable}{lemma}{ScalarRounding}
    \label{lem:convex-represenation-scalar-rounding}
    Let $\ell\in\N_{\ge 0}$ 
    and $y = \sum_{i \in [k]} \lambda_i S_i$ where $\lambda_i \in [0,1]$, $\sum_{i \in [k]} \lambda_i = 1$, and $S_i \in \{0,1\}^E$.
    In polynomial time, we can compute $y' = \sum_{i \in [k]} \lambda'_i S_i$ where $\sum_{i \in [k]} \lambda'_i = 1$ and
    \begin{align}
        \lambda'_i &\in \tfrac{1}{2^{\ell}}\cdot \ZZ, &&  \forall i \in \set{1, \dotsc, k} \label{eq:integer-multiples}\\
        \lambda'_i &= \lambda_i, &&  \forall i \in \set{1, \dotsc, k}, \lambda_i\in\{0,1\} \\
        |y_e - y_e'| &\le k \cdot \tfrac{1}{2^\ell}, &&  \forall e \in E \label{eq:close-solution}\\
        \cost{y'} &\le \cost{y}. \label{eq:rounding-cost-preservation}
    \end{align}
\end{restatable}
We prove this lemma by constructing a flow network and
the standard argument that integral capacities imply existence of an integral min-cost circulation.

Notably, this is the first time we incur a small error for the linear functions $a_j$ while the cost is preserved.
More precisely, we use the lemma with 
$\ell := 2\log(2k)$. From \cref{eq:close-solution} follows that for all~$e \in E$
\begin{equation}
    \label{eq:rdr-scalar-rounding-error}
    |y_e - y_e'| \le k \cdot \tfrac{1}{2^\ell} \le \tfrac{1}{2k}.
\end{equation} 
Therefore, the linear functions also slightly change.
Using \cref{eq:sos-bounded-support,eq:rdr-scalar-rounding-error}, it holds that for all~$j\in \set{1,\dotsc,k}$
\begin{equation}
    \label{eq:rdr-increase-in-linear-function}
    |a_j\T x - a_j\T y'| = |a_j\T y - a_j\T y'| = \sum_{e \in E} (a_j)_e | y_e' - y_e | \le 1.
\end{equation} 
Since $y'$ is a convex combination of (integral)
degree preserving sets in $P$, we have $y'\in P$. In other words, the scalar rounding in \cref{lem:convex-represenation-scalar-rounding} does in fact preserve the degree of~$y$. 

Next, we construct a complete binary tree~$\tree$ with levels~$0,1,\dotsc,\ell$, where each node will be labeled with an edge set.
When the algorithm finishes, the label of the root will be~$X \in \set{0,1}^E$ and satisfy the properties stated in \cref{thm:assign}.
In the following, we describe how the algorithm \textsc{TreeMerge} creates the labels on~$\tree$.
The lowest level~$\ell$ represents the fractional assignment~$y' = \sum_{i \in [k]} \lambda'_i S_i$ where~$S_i \in P\cap \set{0,1}^E$ are degree preserving edge sets.
As we can write~$\lambda'_i = h_i/2^{\ell}$ 
for some $h_i\in\ZZ$ and all $h_i$ sum to $2^{\ell}$, we can naturally label~$h_i$ leaves of level~$\ell$ with $S_i$ for all~$i \in \set{1, \dotsc, k}$.
Thus, $y'$ is the average of all labels of level~$\ell$.
For all $j \in \set{0, \dotsc, \ell-1}$, the labels of level~$j$ are derived from those in level~$j + 1$ such that each node's label is closely related to
those of its two children.
Similar to the last level, each level~$j$ represents a (fractional) edge set~$y_j'$ by taking the average of all labels in this level.

One of the central goals in the construction of $y'_j$ is to guarantee~$\cost y_{j}' \le \cost y_{j+1}'$.
We achieve this by creating two complementary
labelings of level $j$ and selecting the better
of the two.
Denote the labels of level~$j+1$ by~$S_{2i-1}, S_{2i}$ for all~$i\in\set{1, \dotsc, 2^{j}}$. 
Here, each pair~$S_{2i-1}, S_{2i}$ represents the children of the $i$-th node in level $j$. For node~$i$, we construct two potential labels~$T_{i}, T_{i}' \in \set{0,1}^E$ using a random procedure with the following guarantees.

\begin{restatable}{lemma}{EdgeSetDecomposition}
    \label{lem:decomposition}
    Let $S_1, S_2 \in \set{0,1}^E$.
    There exists a random polynomial time procedure that constructs two random edge sets $T, T' \in \set{0,1}^E$ with the following properties. 
    \begin{itemize}
        \item It holds that $T + T' = S_1 + S_2$ and $\E(T) = \E(T') = (S_1+S_2)/2$.
        \item For all $v \in A\cup B$ it holds that $T(\delta(v)), T'(\delta(v)) \in \set{\floor{\,(S_1(\delta(v)) + S_2(\delta(v)))/2 \,}, \ceil{\,(S_1(\delta(v)) + S_2(\delta(v)))/2 \,}}$.
        \item For all $v \in A\cup B$ and all $e \in \delta(v)$, there is at most one edge $e' \in \delta(v)\setminus\{e\}$ such that $T_e$ depends on $T_{e'}$. Likewise, there is at most one edge $e' \in \delta(v)\setminus\{e\}$ such that $T'_e$ depends on $T'_{e'}$. 
    \end{itemize}
\end{restatable}
\medskip
\noindent The lemma can be derived in two steps. First, decompose the symmetric difference of~$S_1$ and~$S_2$ into cycles and paths.
Second, for each of cycle and path, randomly select one of the alternating edge sets for~$T$ and the other for~$T'$.
This random process is similar to other dependent rounding approaches, e.g.~\cite{ChekuriVZ10,GandhiKPS06}, except that we also store~$T'$ that contains the ``opposite'' to every decision in $T$.

We create one fractional assignment from the random edge sets~$T_{i}$, $i\in\{1,\dotsc,2^j\}$,
and one from $T'_{i}$, $i\in\{1,\dotsc,2^j\}$, and pick the lower cost assignment for~$y_{j}'$.
Formally, let 
\begin{equation*}
    z_{j} = \tfrac{1}{2^{j}} \sum_{i \in [2^j]} T_{i} \qquad \text{ and } \qquad
    z'_{j} = \tfrac{1}{2^{j}} \sum_{i \in [2^j]} T'_{i} 
\end{equation*}
From the fact that~$T_{i} + T'_{i} = S_{2i-1} + S_{2i}$, it immediately follows that $(z_{j} + z'_{j})/2 = y_{j+1}'$.
If~$\cost z_j \le \cost z'_j$, set~$y_{j}' = z_{j}$. Otherwise,~$y_{j}' = z'_{j}$.
Consequently, we have that
\begin{equation}
    \label{eq:rdr-cost-of-level-j-assignment}
    \cost y_{j}' \le (\cost z_{j} + \cost z'_{j})/2 = \cost y_{j+1}'.
\end{equation}
We determine the labels of level~$j$ by picking either $T_i$ (if $z_j$ was chosen) or $T'_i$ (if $z'_j$ was chosen).
Repeating the procedure for all~$j \in \set{\ell-1,\dotsc,0}$
results in a label for the root node that is identical to $y_0'$. We conclude by setting $X = y_0'$.
As a last step before proving the main theorem, we 
bound how much the linear functions $a_j$ can change in each level.

\begin{restatable}{lemma}{IncreaseOfLinearFunction}
    \label{lem:linear-function-increase-per-level}
    Let $v\in A\cup B$ and $a \in [0, 1]^E$ with $\supp(a) \subseteq \delta(v)$.
    Let $y_{j+1}' \in [0,1]^E$
    be the fractional solution of the $(j+1)$-th
    level of \textsc{TreeMerge}
    and $y'_{j}\in [0,1]^E$
    that of the $j$-th level.
    Let $t= 132 \ln k$.
    Then with probability at least $1 - 1/k^{10}$, it holds that
    \begin{equation}
        \label{eq:lfipl-bound-on-lf-increase}
        |a\T y_j' - a\T y_{j+1}'| \le 2^{-j/2} \left(t + \sqrt{a\T y_{j+1}' \cdot t}\,\right).
    \end{equation}
\end{restatable}
This lemma follows from a standard Chernoff bound.
We are now in the position to prove \cref{thm:assign} using the lemmas above.

\begin{proof}[Proof of \cref{thm:assign}]
    Let~$\ell = 2 \log (2k)$.
    As explained throughout \cref{sec:dependent-randomized-rounding}, we use \cref{lem:bounded-number-of-frac-variables,lem:convex-represenation-scalar-rounding} to obtain a fractional degree preserving assignment~$y'\in P$ with
    $\cost y'\le \cost x$.
    Note that the rounding of~$x$ to~$y'$ marginally changes the linear function values, but we are able to maintain~$|a_j\T y' - a_j\T x|\le 1$, see \cref{eq:rdr-increase-in-linear-function}.
    Afterwards, we use the algorithm \textsc{TreeMerge} to construct a complete binary tree~$\tree$ with~$\ell + 1$ levels
    and corresponding fractional solutions~$y'_{\ell+1} = y', y'_{\ell}, \dotsc, y'_0 = X$.
    It remains to show that~$X$ satisfies all three properties from the theorem.
    By construction, more precisely \cref{eq:rdr-cost-of-level-j-assignment}, it
    holds that
    \begin{equation*}
        \cost X = \cost y_0' \le \cdots \le y'_{\ell+1} = \cost y' \le \cost x.
    \end{equation*}
    Thus, the cost is preserved.
    Next, we will show that the rounding of~$x$ to~$X$ also preserves the degree.
    Due to \cref{lem:convex-hull-of-edge-preserving-edge-sets,lem:bounded-number-of-frac-variables,lem:convex-represenation-scalar-rounding},
    we have that $y' = \sum_{i \in [k]} \lambda'_i S_i$, where $S_i\in P\cap \{0,1\}^E$
    form the labels of level $\ell$ of $\tree$. For all $i\in\{1,\dotsc,k\}$,
    the fact that $S_i\in P$ implies
    \begin{align}
        \label{eq:rdr-degree-preservation}
        \sss{S_i}{\delta(v)} \in \{ \floor{\,\sss{x}{\delta(v)}\,}, \ceil{\, \sss{x}{\delta(v)}\,} \}.
    \end{align}
    By induction over the tree $\tree$, we show that all labels and, in particular, $X$ are indeed degree preserving. \Cref{eq:rdr-degree-preservation} proves the base case.
    Let $S_1$, $S_2$ be the labels for two children of some node in $\tree$
    and $T, T'$ be the two potential labels for the said node (derived using \cref{lem:decomposition}).
    From the third property of \cref{lem:decomposition} directly follows that for all~$v \in A \cup B$ 
    \begin{equation}
        \label{eq:rdr-degree-preservation-of-edge-decomposition}
        T(\delta(v)), T'(\delta(v))
        \in \big\{\floor{ \,\tfrac{1}{2}  \sss{S_{1}}{\delta(v)} + \tfrac{1}{2}\sss{S_{2}}{\delta(v)} \,}, \ceil{\,\tfrac{1}{2} \sss{S_{1}}{\delta(v)} + \tfrac{1}{2}\sss{S_{2}}{\delta(v)}\,} \big\}.
    \end{equation}
    By induction hypothesis, $S_1$ and $S_2$ are degree preserving, so
    \begin{align*}
        &\floor{ \,\tfrac{1}{2}  \sss{S_{1}}{\delta(v)} + \tfrac{1}{2}\sss{S_{2}}{\delta(v)} \,} \ge \floor{ x(\delta(v)) } \text{ and similarly} \\
        &\ceil{ \,\tfrac{1}{2}  \sss{S_{1}}{\delta(v)} + \tfrac{1}{2}\sss{S_{2}}{\delta(v)} \,} \le \ceil{ x(\delta(v)) } .
    \end{align*}
    This concludes the induction step.
    
    It remains to prove the concentration, i.e., that $X$ marginally deviates from~$x$ in each of the
    given linear functions $a_j$.
    We apply \cref{lem:linear-function-increase-per-level} together with a union bound over all $k$ linear functions and all $\ell = 2\log(2k) \le k$ levels of $\tree$.
    Let $t = 30\log k$.
    As a consequence, with probability at least~$1 - 1/k^8$, it holds for all levels~$i \in \set{1, \dotsc , \ell}$ and linear functions~$a_j,j\in\set{1,\dotsc,k}$ that
    \begin{equation}\label{eq:change-a}
        |a_i\T y_{i+1}' - a_i\T y_i'| \le 2^{-i/2} \Big(t + \sqrt{a\T y_{i+1}' \cdot t}\,\Big).
    \end{equation}
    For some universal constant $d$, we prove that for all $j\in\{1,\dotsc,k\}$ and all~$i\in\{0,\dotsc,\ell\}$
    \begin{equation}\label{eq:grow-a}
        a_j\T y'_i \le d (1 + t + a\T x).
    \end{equation}
    Let us first argue that this in fact implies the last part of the theorem.
    Using triangle inequality and geometric series, it holds that
    \begin{align*}
        |a_j\T X - a_j\T x| 
        = |a_j\T y'_{\ell+1} - a_j\T y'_{\ell}| 
        &\le \sum_{i \in \set{0, \dotsc, \ell-1}} |a_j\T y'_{i+1} - a_j\T y'_{i}| \\
        &\le \sum_{i \in \set{0, \dotsc, \ell-1}} \frac{1}{(\sqrt{2})^i} \Big(t + \sqrt{a\T y'_{i+1} \cdot t}\,\Big) \\
        &\le \sum_{i \in \set{0, \dotsc, \ell-1}} \frac{1}{(\sqrt{2})^i} \Big(t + \sqrt{d(1 + t + a\T x) \cdot t}\,\Big) \\
        &= O(\max\{\log k, \sqrt{a\T x \cdot \log k}\}) .
    \end{align*}
    Finally, we prove \cref{eq:grow-a}.
    Let $j\in\{1,\dotsc,k\}$
    and $i\in\{0,\dotsc,\ell-1\}$.
    Let $i' \ge i$ be the minimal index such that
    $a_j\T y'_{i'} \le 1 + t + a_j\T x$.
    As $a_j\T y'_{\ell} = a_j\T y' \le 1 + a_j\T x$, such~$i'$ must exist.
    If $i' = i$, we are done. If $i' \neq i$, then it follows from \cref{eq:change-a} that
    \begin{equation*}
        a\T_j y'_{i'-1} \le (1 + t + a_j\T x) + 2^{-(i'-1)/2}\big(t + (1 + t + a_j\T x)\big) \le 3(1 + t + a_j\T x) \ .
    \end{equation*}
    For all $i'' \in\{i'-1,\dotsc,i\}$, we have $a_j\T y'_{i''} > t$. Thus, the same equation implies
    \begin{equation*}
        a\T_j y'_{i''-1} \le a\T_j y'_{i''} + \frac{2}{2^{(i''-1)/2}} a\T_j y'_{i''} = \left(1 + \frac{1}{2^{(i''-3)/2}}\right) a\T_jy'_{i''} \ .
    \end{equation*}
    Using the inequality $1 + z \le \exp(z)$ for all $z\in\R$, we have
    \begin{align*}
        a\T_j y'_{i} 
        &\le a_j\T y'_{i' - 1} \cdot \prod_{i''=i'-1}^{i+1} \left(1 + \frac{1}{2^{(i''-3)/2}}\right) \\
        &\le a_j\T y'_{i' - 1} \cdot \exp\left(\sum_{i''=i'-1}^{i+1} \frac{1}{2^{(i''-3)/2}}\right) \\
        &\le 3 e^{O(1)} \cdot (1 + t + a_j\T x) .
    \end{align*}
    This shows \cref{eq:grow-a} and thereby concludes the proof.
\end{proof}

\subsection{Omitted Proofs of Dependent Rounding}
\label{subsec:appendix-ommitted-proofs-of-dependent-rounding}

In this section, we provide the proofs of the previously stated lemmas.
Recall that $P$ is the convex hull of integral degree preserving edge sets for $x\in [0,1]^E$.
Let the operator $\oplus$ denote the symmetric difference (i.e., XOR) of two edge sets.
We start by showing that~$P$ indeed contains $x$.

\begin{lemma}
    \label{lem:convex-hull-of-edge-preserving-edge-sets}
    Let $q = |E|$.
    There exists a convex representation $x = \sum_{i \in [q]} \lambda_i S_i \in P$ where $S_i \in P\cap\{0,1\}^E$ and $\lambda_i \in [0,1]$ and
    $\sum_{i \in [q]} \lambda_i = 1$ such that for all $v\in A\cup B$ and~$i \in \set{1,\dotsc,q}$ holds that $\sss{S_i}{\delta(v)} \in \{\floor{\,\sss{x}{\delta(v)}\,}, \ceil{ \,\sss{x}{\delta(v)} \,}\}$.
\end{lemma}

\begin{proof}
    We rely on a standard flow argument. 
    To this end, construct a digraph $D_f = (V_f , A_f)$ as follows. Let $V_f = \{s,t\} \cup V$ be the set of vertices, $A_s$ be a set of arcs directed from $s$ to each $a \in A$, $A_t$ be the set of arcs directed from each $b \in B$ to $t$, and $A_E$ be a directed variant of $E$ from $A$ to $B$. Let $A_f = A_s \cup A_E  \cup A_t \cup \{(t,s)\}$ be the set of arcs in $D_f$. Set the capacity interval for arc $(s,a) \in A_s$ as $[\lfloor x(\delta(a)) \rfloor, \lceil x(\delta(a)) \rceil]$ and similarly for each arc $(b,t) \in A_t$ as $[\lfloor x(\delta(b)) \rfloor, \lceil x(\delta(b)) \rceil]$. Moreover, set the capacity interval for each arc $e \in E$ to $[0,1]$ and for the arc $(t,s)$ to $[0,\infty]$. A function $f: A_f \rightarrow \R$ is called a circulation if $f(\delta^{\mathrm{in}}(v)) = f(\delta^{\mathrm{out}}(v))$ for each vertex $v \in V_f$. 
    
    One can naturally derive a feasible fractional circulation from the vector $x\in [0,1]^E$: the flow from $s$ to $a \in A$ is $x(\delta(a))$, which lies in $[\lfloor x(\delta(a)) \rfloor, \lceil x(\delta(a)) \rceil]$, the circulation from $b \in B$ to $t$ is $x(\delta(b))$, which lies in $[\lfloor x(\delta(b)) \rfloor, \lceil x(\delta(b)) \rceil]$, the circulation on each arc $(a,b) \in A_E$ is $x_{ab}$, and the circulation on arc $(t,s)$ is $x(E)$. Hence, $x$ satisfies the capacity and flow conservation constraints. 
    It is well known, see e.g.~\cite[Corollary 13.10b]{schrijver2003combinatorial}, that for integral capacities the set of all feasible circulations (including $x$) forms an integral polytope. Thus, we can rewrite the circulation above as a convex combination
    of integral circulations. Each of these integral circulations corresponds
    to a degree preserving edge set.
\end{proof}

Further, there always exists a comparable convex representation which has only few fractional variables in the support of each constraint.

\BoundedNumberOfFracVariables*

\begin{proof}
    First, we argue about the structure of the edges of polytope $P$.
    Let $S, T\in \{0,1\}^E$ be the two vertices at the two ends of
    some edge of $P$. We claim that $S\oplus T$ is a simple cycle
    or a simple path. Suppose not. If $S\oplus T$ is acyclic, let $D$
    be a maximal path in $S\oplus T$; otherwise let $D$ be a simple cycle
    contained in $S\oplus T$.
    Note that $S \oplus D \in P$ and $T \oplus D\in P$ and both
    points do not lie on the edge between $S$ and $T$.
    However, $(S + T)/2 = (T \oplus D + S \oplus D) / 2$ contradicts
    that $S$ and $T$ are connected by an edge.

    Let $y$ be an optimal vertex solution of the linear program
    \begin{align*}
        \min c\T& y \\
        a_j\T y &= a_j\T x &\forall j \in \set{1,\dotsc,k} \\
        y &\in P
    \end{align*}
    Due to~\cref{lem:convex-hull-of-edge-preserving-edge-sets}, the linear program is feasible.
    The solution $y$ must lie on a face~$F$ of~$P$ with dimension at most $k$.
    Consider an arbitrary vertex $S$ of $F$. Furthermore, let $T_1,\dotsc,T_h$, $h\le k$, be the vertices of $F$ such that there is an edge between each $T_i$ and $S$.
    Thus, we can write $y = S + \sum_{i \in [h]} \lambda_i (T_i - S)$ where $\lambda_1,\dotsc,\lambda_h \ge 0$.
    By our previous argument, $T_i \oplus S$ is a simple cycle or path for each $i\in \set{1,\dotsc,h}$.
    In particular, $|(T_i \oplus S) \cap \delta(v)| \le 2$ for each $v\in A\cup B$. This implies that the last property holds for $y$.
\end{proof}

Allowing a small rounding error, there always exists a convex representation where the scalars are integer multiples of a power of two.

\ScalarRounding*

\begin{proof}
    Using a standard flow argument,
    we construct a digraph $D_f = (V_f, A_f)$ that represents a circulation network. For each arc in $a \in A_f$, we adjust a capacity interval such that every feasible circulation corresponds to a solution that is ``close'' to $x$ and preserves the cost. Let $V_f =\{ s, t \} \cup \{ v_i : i\in [k]\}$ where nodes $v_i$ correspond to each $S_i$. The set of arcs~$A_f$ contains arcs $(t, s)$, $(s,v_i)$, and $(v_i,t)$ for $i \in [k]$. A function $f: A_f \rightarrow \R$ is called a circulation if $f(\delta^{\mathrm{in}}(v)) = f(\delta^{\mathrm{out}}(v))$ for each vertex $v \in V_f$. 
    Set the capacity interval of arc~$(t,s)$ to~$[2^\ell, 2^\ell]$ and
    the one of arcs~$(s,v_i)$ and~$(v_i, t)$ to $[\lfloor \lambda_i 2^\ell \rfloor, \lceil \lambda_i 2^\ell \rceil]$ for $i \in [k]$. 
    Moreover, define a linear cost function $\mathrm{cost}(s,v_i) = \sum_{e \in E} c_e (S_i)_e$, the contribution of $S_i$ to the total cost.
    Set $\mathrm{cost}(t,s) = \mathrm{cost}(v_i, t) = 0$.
    The scalars $\lambda_i$ induce a natural circulation $f$: on 
    arcs~$(s, v_i)$ and~$(v_i, t)$, we send a flow of $\lambda_i 2^\ell$
    and on arc~$(t, s)$, we send a flow of $2^\ell$.
    
    For integral capacities the set of feasible circulations forms an integral polytope, see e.g.~\cite[Corollary 13.10b]{schrijver2003combinatorial}.
    Thus, there exists a minimum cost integral circulation. Let $\bar{f}$ be the this integral circulation. We define $\lambda'_i = \bar{f}(s, v_i) / 2^\ell$ for each $i\in\{1,\dotsc,k\}$, then
    \begin{align*}
    \sum_{i \in [k]} \lambda'_i = \sum_{i \in [k]} \frac{ \bar{f}(s,v_i)}{2^\ell} = \frac{1}{2^\ell} \cdot \bar{f}(t, s) = \frac{1}{2^\ell} \cdot 2^\ell  = 1.
    \end{align*}
    Consequently, $\lambda'_i$ creates a valid convex combination for $y'$ where each $\lambda'_i$ is a multiple of~$1/2^\ell$. For each $e \in E$
    \begin{align*}
    |y_e - y'_e| =  \left| \sum_{i \in [k]} \lambda_i (S_i)_e - \sum_{i \in [k]} \lambda'_i (S_i)_e \right| =  \sum_{i \in [k]} |\lambda_i - \lambda'_i| \cdot (S_i)_e \leq \sum_{i \in [k]} \frac{1}{2^\ell} = k \cdot \frac{1}{2^\ell}.
    \end{align*}
    Since the integer circulation minimizes the total cost, we have
    \begin{equation}
        \label{eq:psr-total-rounded-cost}
        \cost{y'} = \sum_{e\in E} c_e \sum_{i \in [k]} \lambda'_i (S_i)_e = \frac{1}{2^\ell} \sum_{e \in E} \bar{f}(e) \cdot \mathrm{cost}(e)
        \le \frac{1}{2^\ell} \sum_{e \in E} f(e) \cdot \mathrm{cost}(e)
        = \cost{x}.
    \end{equation}
    Moreover, constructing $D_f$ and solving minimum cost flow can be done in polynomial time~\cite{schrijver2003combinatorial}.
\end{proof}

The \textsc{TreeMerge} algorithm described in \cref{sec:dependent-randomized-rounding} uses the following lemma to construct new random edge sets while preserving the degree of the former sets.

\EdgeSetDecomposition*

\begin{proof}
    For the randomized construction of $T$ and $T'$, we distinguish whether an edge is present in the symmetric difference~$S_1 \symdif S_2$ or not.
    For any edge~$e \in E$ where $(S_1)_e = (S_2)_e$, we set~$T_e = \big((S_1)_e + (S_2)_e \big)/2$ and~$T_e' = \big((S_1)_e + (S_2)_e \big)/2$.
    Next, partition the edges in~$S_1 \symdif S_2$ into simple paths and cycles with $E_1\dot\cup \cdots \dot\cup E_k = S_1 \symdif S_2$ with the following property: each odd-degree vertex is the endpoint of exactly one path and no path ends in an even-degree vertex.
    This easily follows from iteratively removing cycles and maximal paths.

    Let $i\in\set{1,\dotsc,k}$. Choose $C_i \dot\cup D_i = E_i$ such that no two edges in $C_i$ or in $D_i$ are adjacent. This is possible by alternatingly assigning edges to $C_i$ and $D_i$, as $E_i$ is an even length cycle or a path.
    We make a uniform binary random decision $R_{i} \in \set{0,1}$, which is independent of all $R_{i'}$, $i'\neq i$.
    This random variable indicates whether~$C_i$ is assigned to~$T$ or~$T'$ ($D_i$ is then assigned to the other edge set).
    More precisely, for each~$e \in E_i$ we set 
    \begin{equation}
        \label{eq:mom-random-matching-choice-for-pair}
        T_e = 
        \begin{cases}
            (C_i)_e & \text{if } R_{i} = 0, \\
            (D_i)_e & \text{if } R_{i} = 1,
        \end{cases}
        \qquad \text{and} \qquad 
        T_e' = 
        \begin{cases}
            (D_i)_e & \text{if } R_{i} = 0, \\
            (C_i)_e & \text{if } R_{i} = 1.
        \end{cases}
    \end{equation}
    In particular, $T_e + T_e' = (C_i)_e + (D_i)_e = 1 = (S_1)_e + (S_2)_e$.
    Since both outcomes~$R_i = 1$ and~$R_i = 0$ have probability $1/2$, it holds that $\E(T_e) = \E(T_e) = \big((C_i)_e + (D_i)_e \big)/2 = \big((S_1)_e + (S_2)_e \big)/2$ for all~$e \in S \symdif T$.
    Thus, it follows that $\E(T) = \E(T') = (S_1+S_2)/2$.
    Let~$\F = \set{\delta(v) \mid v \in A \cup B}$.
    Considering each vertex in a simple path or cycle is incident to at most two edges, we have~$| F \cap E_i| \le 2$ for all~$F \in \F$ and~$i \in \set{1, \dotsc, k}$.
    As the random choice for each path and cycle is independent, it holds that for all~$F \in \F$ and~$e \in F$, there is at most one edge~$e' \in F$ with~$e \neq e'$ such that~$T_e$ depends on~$T_{e'}$ (and the same for~$T'$).
    
    It remains to show that~$T(F), T'(F) \in \set{\floor{\,(S_1(F) + S_2(F))/2\,}, \ceil{\,(S_1(F) + S_2(F))(2)\,}}$.
    To this end, we distinguish whether a vertex $v \in A \cup B$ has even or odd degree. If $|\delta(v)|$ is even, then $v$ is always incident to exactly one edge in $C_i$ and one in $D_i$, for each $i\in\set{1,\dotsc,k}$.
    Hence, $T(F) = T'(F) = |\delta(v)| / 2 = (S_1(F) + S_2(F))/2$.
    If $|\delta(v)|$ is odd, then there is exactly one path that ends in $v$ due to the choice of decomposition into cycles and paths. 
    Consequently, the arguments above apply to all but one edge in $\delta(v)$ and therefore we have $|T(F) - T'(F)| = 1$. 
    This implies the claim.
\end{proof}

The last lemma bounds the change of a linear function between two consecutive levels of~\textsc{TreeMerge}.

\IncreaseOfLinearFunction*

\begin{proof}
    Let $T_i$, $i \in [2^{j}]$, be the random edge set created on the $j$-th level from the edge sets~$S_{2i-1}$,$S_{2i}$, $i\in [2^{j}]$ of the $(j+1)$-th level.
    Recall that the procedure from \cref{lem:decomposition} actually creates
    two alternative edge sets~$T_i,T'_i$ of which \textsc{TreeMerge} selects just one.
    However, the solution derived from $T_i$ satisfies \cref{eq:lfipl-bound-on-lf-increase} if and only if the one from $T'_i$ does. Hence, it suffices to show
    it for one of the two.
    Further, for the sake of convenience, we analyze the scaled expression~$|2^j \cdot a\T y'_j - 2^j \cdot a\T y'_{j+1}|$ instead of~$|a\T y'_j - a\T y'_{j+1}|$.
    Due to \cref{lem:decomposition}, it holds that $\E(T_i) = (S_{2i-1} + S_{2i})/2$.
    Thus,
    \begin{align*}
        \label{eq:ilf-expected-value}
        \E[2^j \cdot a\T y_{j}'] 
        = a\T \Big(\sum_{i \in [2^j]} \E[T_i] \Big)
        = \tfrac{1}{2} a\T \Big(\sum_{i \in [2^{j+1}]} S_{i} \Big)
        = 2^j \cdot a\T y_{j+1}'.
    \end{align*}
    Since \textsc{TreeMerge} independently constructs the~$T_i$ on the $j$-th level, any two random edge sets~$T_i$ and~$T_{i+1}$ are independent.
    Moreover, it follows from \cref{lem:decomposition} that for all~$T_i$ and~$e \in \delta(v)$, there exists at most one edge~$e' \in \delta(v)$ such that $(T_i)_e$ depends on~$(T_i)_{e'}$.
    Consequently, there exists a partition of $\delta(v)$ into~$P_1$ and~$P_2$ such that all variables $(T_i)_e$, $e\in P_1$, as well as~$(T_i)_e$, $e\in P_2$ are independent.
    Let~$u = \sum_{i \in [2^{j}]} U_{i} \in [0,1]^E$ and~$w = \sum_{i \in [2^{j}]} W_{i} \in [0,1]^E$ where~$U_i, W_i \in \{0,1\}^E$ are defined for all~$e \in E$ as
    \begin{align*}
        (U_i)_e =
        \begin{cases}
            1, & \text{if } e\in\delta(v) \text{ and } (T_i)_e \in P_1 \\
            0, & \text{otherwise}
        \end{cases}
        \; \text{and} \quad
        (W_i)_e =
        \begin{cases}
            1, & \text{if } e\in\delta(v) \text{ and } (T_i)_e \in P_2 \\
            0, & \text{otherwise}.
        \end{cases}
    \end{align*}
    Hence, we have~$2^j \cdot a\T y_{j}' = a\T(u + w)$.
    Next, we show that
    \begin{equation}\label{eq:chernoff-one}
        \Pr \left[ | a\T u - \E[a\T u] | > \tfrac{1}{2} t + \tfrac{1}{2} \sqrt{\E[a\T u]  \cdot t} \right] < \tfrac{1}{2 k^{10}}.
    \end{equation}
    For brevity, define $\mu = \E[a\T u]$.
    We distinguish two cases. If~$t > 4\mu$, then set~$\delta = t / (4\mu) > 1$.
    Notice that $a\T u$ is the sum of independent variables, each contained in $[0,1]$.
    We have
    \begin{align*}
        \Pr \left[ a\T u < \mu - t/4 \right] &= 0
    \end{align*}
    It follows from a standard Chernoff bound that
    \begin{align*}
        \Pr \left[ a\T u  > \mu + t/4 \right] 
        &= \Pr \left[ a\T u > (1 + \delta) \mu \right] \\
        &\le \exp(-\tfrac{1}{3} \delta \mu) \\
        &= \exp(- 11 \ln k) \\
        &\le \tfrac{1}{2 k^{10}}.
    \end{align*}
    If~$t \le 4\mu$, then set~$\delta= 1/2 \cdot \sqrt{t/\mu} \in (0,1]$. Again by a Chernoff bound, we have
    \begin{align*}
        \Pr \left[ |a\T u  - \mu| > \tfrac{1}{2}\sqrt{\mu t} \,\right]
        &= \Pr \left[ |a\T u - \mu| > \delta \mu \right] \\
        &\le 2 \exp(- \tfrac{1}{3} \delta^2 \mu) \\
        &= 2 \exp(- 11 \ln k) \le \tfrac{1}{2 k^{10}} .
    \end{align*}
    By symmetry, \cref{eq:chernoff-one} holds for $w$ as well. 
    We conclude that with probability at least~$1 - 1/k^{10}$, we have
    \begin{align*}
        | 2^j a\T y'_j - 2^j a\T y'_{j+1} | 
        &\le | a\T u - \E[a\T u] | + | a\T w - \E[a\T w] | \\ 
        &\le \tfrac{1}{2} t + \tfrac{1}{2} \sqrt{\E[a\T u]  \cdot t} + \tfrac{1}{2} t + \tfrac{1}{2} \sqrt{\E[a\T w] \cdot t} \\
        &\le t + \sqrt{\E[2^j \cdot a\T y'_j] \cdot t} \\
        &\le 2^{j/2} \big(t + \sqrt{a\T y'_{j+1} \cdot t}\big). \qedhere
    \end{align*}
\end{proof}
\section{Application to Budgeted Santa Claus Problem}
\label{sec:santa-claus}

In this section, we present our approximation algorithm the Budgeted Santa Claus Problem based on the dependent rounding scheme described in \cref{sec:dependent-randomized-rounding}.

\subsection{Linear Programming Formulation}
\label{subsec:lp}
Introducing an LP relaxation for the Budgeted Santa Claus Problem,  we first reduce the problem to its decision variant.
For a given threshold $T \ge 0$, the goal is to either
find a solution of value $T/ \alpha$ or determine that
$\OPT < T$, where $\OPT$ is the optimal value of the original optimization problem. This variant is equivalent to an
$\alpha$-approximation algorithm by a standard
binary search framework. 

Based on $T$, intuitively thought of as the optimal value, we partition the resources into two sets by size. Set $\cB$ consists of the \emph{big resources} with $\cB \coloneq \{ j \in \cR : v_j \geq T/\alpha\}$ and set~$\cS$ consists of the \emph{small resources} $\cS \coloneq \{ j \in \cR : v_j < T/\alpha\}$.
We use assignment variables that indicate whether a particular
resource is assigned to a particular player. For clarity,
we use different symbols for big and small resources.
Let~$x_{ib} \in [0,1]$ denote the portion of big resource $b \in \cB$ that player $i \in  \cP$ receives. Similarly, denote by $z_{is} \in [0,1]$ the portion of the small resource $s \in \cS$ that player $i$ receives.
Unlike the original problem, we allow these
assignments to be fractional in the relaxation.
Since naive constraints on these variables lead to an unbounded integrality gap, see e.g.~\cite{BansalS06}, we use
non-trivial constraints inspired by an LP formulation of Davies, Rothvoss and Zhang~\cite{DaviesRZ20}. Here, we make the structural
assumption that in any solution, a player either receives exactly
one big resource (and nothing else) or only small resources.
Towards the goal of obtaining a solution of value $T/\alpha$,
any big resource is sufficient for any player and receiving more
would be wasteful.
If there is a solution of value $T$, then there is also
a pseudo-solution such that each player either gets exactly one big
resource (and nothing else) or a value of at least $T$ from
small resources only. Note that it might be that the former
type of player only has a value of $T/\alpha$. Thus, if $T \le \OPT$, then the following relaxation called \LP{T}
is feasible and has a value at most $C$.
\begin{align}
    \min \; \sum_{i \in \cP}^{} &\left[\sum_{b \in \cB}^{}c_{ib}\cdot x_{ib} + \sum_{s \in \cS}^{}c_{is}\cdot z_{is} \right]  && \label{eq:economical-objective-function} \\
    \sum_{s \in \cS} v_{s}\cdot z_{is} &\geq T \cdot \left(1 - \sum_{b \in \cB} x_{ib}\right) &&\forall i \in \cP \label{eq:economical-value} \\
    z_{is} &\leq 1 - \sum_{b \in \cB}^{} x_{ib} &&\forall s \in \cS, i \in \cP \label{eq:economical-small-gifts} \\
    \sum_{i \in \cP}^{} x_{ib} &\leq 1 &&\forall b \in \cB \label{eq:economical-ub-player} \\
    \sum_{i \in \cP}^{} z_{is} &\leq 1 &&\forall s \in \cS \label{eq:economical-ub-z} \\
    \sum_{b \in \cB}^{} x_{ib} &\leq 1 &&\forall i \in \cP \label{eq:economical-ub-big-gifts} \\
    z_{is}, x_{ib}  &\geq 0 &&\forall s \in \cS, b \in \cB, i \in \cP \label{eq:economical-nonneg-x}
\end{align}

The constraints~\eqref{eq:economical-ub-player} and \eqref{eq:economical-ub-z} describe that
each big or small resource is only assigned once.
Justified by earlier arguments, constraint~\eqref{eq:economical-ub-big-gifts} ensures that each player receives at most one big resource.
Considering constraint~\eqref{eq:economical-small-gifts}, we only need to verify that the constraint is valid for integer solutions.
By our
assumption, if player $i$ receives one big resource, then it should not get any small resources,
which is exactly what the constraint expresses. Conversely, if the player does not receive any big resources, the
constraint is trivially satisfied. Similarly, there are two cases for
Constraints~\eqref{eq:economical-value}. If
player~$i$ receives one big resource, the constraint is trivially satisfied. Otherwise, it must receive a value of at least $T$ in small resources. During our rounding procedure in \Cref{sec:rounding-of-small-items}, we essentially lose
some value from small resources and can only guarantee
a value of $T/\alpha$ for players without a big resource.

\subsection{Technical Goals}
Let $(x, z)$ be a feasible solution to \LP{T}, where $x$ and $z$ represent the vectors of assignment variables corresponding to big and small resources, respectively. Formally, we have $x_{ib}, z_{is} \in [0,1]$ for $i \in \cP, b\in \cB, s \in \cS$.
We will define a randomized rounding procedure that constructs
a distribution over the binary variables
$X_{ib} \in \{0,1\}$ and $Z_{is}\in\{0, 1\}$ describing whether a big resource $b \in \cB$ or small resource $s\in \cS$ is assigned to player~$i \in \cP$.
For notational convenience, define $Y_i = 1 - \sum_{b\in B} X_{ib}$ as the indicator variable whether a player~$i$ \emph{does not}
get a big resource (and needs small resource). Similarly, let $y_i = 1 - \sum_{b\in \cB} x_{ib}$ be the corresponding value to~$Y_i$ from the corresponding value from the LP variables.

Our goal is that the total cost of assignments does not exceed the budget $C$ and the integral solution $(X,Z)$ is an $\alpha$-approximation solution with respect to the minimum value a player
receives. In other words, we want to obtain an integral solution $(X,Z)$ that satisfies the following two properties.
\smallskip
\begin{enumerate}
    \item $\sum_{i \in \cP}^{} \left[\sum_{b \in \cB}^{}c_{ib}\cdot X_{ib} + \sum_{s \in \cS}^{}c_{is}\cdot Z_{is} \right] \leq C$.
    
    \item Every player receives resources of value at least $T/\alpha$ with high probability.
\end{enumerate}

We first apply our dependent rounding scheme to round
the assignment of big resources to an integral one.
To cover the players that do not receive big resources, i.e.,
those with~$Y_i = 1$, we need to change the assignment of
small resources as well.
Initially, some small resources will be assigned fractionally
and even more than once.
In a second step, we transform the solution into one
where each small resource is assigned only once---incurring a
loss in the value that the players receive.

\subsection{Rounding of Big Resources}
\label{sec:rounding-of-big-items}
The following lemma summarizes the properties we derive
from the dependent rounding scheme. 
Note that while the assignment of small resources can change, it remains fractional for now.

\begin{lemma}
    \label{lem:cost-preservation}
    Let $(x,z)$ be a feasible solution to \LP{T, C}. There is a randomized algorithm that produces an assignment $X_{ib}\in\{0,1\}$, $i\in \cP$, $b\in\cB$, and $z'_{is}\in [0,1]$, $i\in \cP$, $s\in\cS$ such that
    with high probability
    \begin{enumerate}
        \item $\sum_{i\in\cP}\sum_{b\in\cB} c_{ib}\cdot X_{ib} +  \sum_{i\in\cP }\sum_{s\in\cS} c_{is} z'_{is} \leq C$, \label{eq:costs}
        \item Each big resource is assigned at most once, i.e., $\forall b\in \cB :\sum_{i\in\cP} X_{ib} \le 1$,
        \item Each small resource is assigned at most $O(\log n)$ times, i.e., $\forall s\in \cS: \sum_{i\in\cP} z'_{is} \le O(\log n)$,
        \item Each player receives either one big resource or a value of at least $T$ in small resources.
    \end{enumerate}
\end{lemma}
\begin{proof}
We first describe the instance to which we will apply the dependent rounding procedure from \cref{thm:assign}.
We build a bipartite graph $\Gx{x,z} = (\cP\cup (\cB \cup \{d\}), E)$, where $\cP$ is the set of players, $\cB$ is the set of big resources, and $d$ is a \textit{dummy node}. 
The role of~$d$ can be summarized as: every player who does not get a big resource selects
the edge to~$d$. 
Let graph~$\Gx{x,z}$ contain an edge $(i,b) \in E$ labeled with cost~$c_{ib}$ between each big resource~$b$ and player~$i$ with $x_{ib} > 0$. 
Let $J$ denote the set of all these edges. 
For every player~$i$ with~$y_i > 0$, add an edge~$(i, d)$ to the dummy node of cost~$\sum_{s \in \cS} c_{is} \cdot (z_{is}/y_i) \in \R^n$.
Let $K$ denote the set of these edges and
set $E = J\cup K$.
Intuitively, $x_{ib}$ is a fractional edge selection
of edges in~$J$ and~$y_i$ a fractional edge selection of $K$ (where $y_i$ corresponds to edge $(i, d)\in K$).
This fractional edge set satisfies the following.
\smallskip
\begin{enumerate}
    \item[(i)] For each player $i$, we have $\sum_{b \in \cB} x_{ib} + y_i = 1$ due to the definition of $y_i$.
    \item[(ii)] For each big resource $b$, we have $\sum_{i \in \cP} x_{ib} \leq 1$, due to Constraint~\eqref{eq:economical-ub-player}.
\end{enumerate}

Applying the dependent rounding procedure of \Cref{thm:assign}
with $k = |\cS| = n$ linear functions (we defer a precise definition to the end of the proof),
we obtain an integral edge selection~$X_{ib}\in\{0,1\}$, $i\in \cP$, $b\in \cB$ and $Y_i\in\{0,1\}$, $i\in \cP$.
Here, similar to before, $Y_i$ defines the edge selection in~$K$,
i.e.,~edge $(i, d)$ is selected if $Y_i = 1$.
From~(ii) and \Cref{thm:assign}, it follows that
$\sum_{i\in \cP} X_{ib} \le 1$ for all $b\in \cB$,
which means that Property~2 of the lemma holds.
Further,
\begin{align*}
    \sum_{i\in\cP}\sum_{b\in\cB} c_{ib}\cdot X_{ib} +  \sum_{i\in\cP : y_i > 0} Y_i \cdot (\sum_{s\in\cS} c_{is}\cdot \frac{z_{is}}{y_i}) 
    &\le \sum_{i\in\cP}\sum_{b\in\cB} c_{ib}\cdot x_{ib} +  \sum_{i\in\cP : y_i > 0}(\sum_{s\in\cS} c_{is}\cdot \frac{z_{is}}{y_i}) \cdot y_i\\
    &\le \sum_{i\in\cP}\sum_{b\in\cB} c_{ib}\cdot x_{ib} +  \sum_{i\in\cP}\sum_{s\in\cS} c_{is}\cdot z_{is} \\
    &\leq C,
\end{align*}
where the first inequality follows from the cost preservation
property of \Cref{thm:assign} and the last inequality from the feasibility of~$(x,z)$. 
Notably, if $y_i = 0$, then $Y_i = 0$ as
there does not exist an edge $(i, d)$.
In particular, $Y_i = 1$ implies that $y_i > 0$.
We define the assignment of small resources as 
\begin{equation*}
z'_{is} = \begin{cases}
     z_{is}/y_i &\text{ if } Y_i = 1,\\
     0 &\text{ otherwise.}
\end{cases}
\end{equation*}
Due to Constraint~\eqref{eq:economical-small-gifts}, we have $z_{is} \le y_i$ and hence $z'_{is} \in [0, 1]$. Consequently, Property~1
immediately follows from the previous cost calculation and the definition of $z'_{is}$.
From~(i) and \Cref{thm:assign}, we obtain $\sum_{b\in \cB} X_{ib} + Y_i = 1$ for each $i\in \cP$.
In words, player~$i$ either gets a big resource or $Y_i = 1$.
In the latter case holds that
\begin{equation*}
    \sum_{s\in \cS} v_s z'_{is} = 
    \sum_{s\in \cS} v_s z_{is} / y_i \ge T .
\end{equation*}
Here, the inequality follows from the definition of $y_i$ 
and Constraint~\eqref{eq:economical-value}.
Therefore, Property~4 holds.

It remains to show Property~3, i.e, that every small resource
is assigned at most $O(\log n)$ times.
To this end, we define the linear functions provided to \Cref{thm:assign}: 
for each small resource $s\in \cS$,
there is one linear function~$a^{(s)}\in [0,1]^K$,
i.e., specified by the edges incident to~$d$.
Implicitly, the coefficient for all other edges is zero.
Thus,
the linear function satisfies the support restriction of 
\Cref{thm:assign}.
For $e = (i, d) \in K$, we define
$a^{(s)}_e = z_{is} / y_i$.
Using Constraint~\eqref{eq:economical-small-gifts}, it holds that
\begin{equation*}
\sum_{e = (i, d) \in K} y_i \cdot a^{(s)}_e = 
\sum_{e = (i, d) \in K} y_i \cdot \frac{z_{is}}{y_i} = 
\sum_{e = (i, d) \in K} z_{is} \le 1 . 
\end{equation*}
Finally, from the concentration bound of \Cref{thm:assign} follows
\begin{equation*}
\sum_{i\in \cP} z'_{is}
= \sum_{e = (i, d) \in K} Y_i \cdot z'_{is}
= \sum_{e = (i, d) \in K} Y_i \cdot z_{is} / y_i
= \sum_{e = (i, d) \in K} Y_i \cdot a^{(s)}_e \le O(\log n) . \qedhere
\end{equation*}
\end{proof}

\subsection{Rounding of Small Resources}
\label{sec:rounding-of-small-items}
In the previous subsection, we described how big resources $b \in \cB$ were integrally assigned to the players.
As some players did not receive any big resources and still need to be covered by small resources, let
$\cQ$ be the set of those players, i.e., players~$i\in\cP$ 
for which $Y_i = 1$. The linear program in the following lemma
corresponds to the property of the assignment variables
for small resources from the previous section.

\begin{restatable}{lemma}{RoundingSmallItems}
    \label{lem:rounding-small-items}
    Let $\cQ \subseteq \cP$ and consider the LP $\mathrm{small}(T, \beta)$ defined as
    \begin{align*}
        \sum_{s\in \cS} v_s \cdot z'_{is} &\ge T &\forall i\in\cP \\
        \sum_{i\in \cQ} z'_{is} &\le \beta &\forall s\in\cS\\
        z'_{is} &\ge 0 &\forall i\in \cQ, s\in \cS
    \end{align*}
    If $\mathrm{small}(T, \beta)$ has a fractional solution $z'$,
    then $\mathrm{small}(T/\beta - \max_{s\in\cS} v_s, 1)$ has an integral solution $Z$ that can be found in polynomial time with
    \begin{equation*}
       \sum_{s\in\cS}\sum_{i\in\cQ} c_{is}\cdot Z_{is}     
       \le \sum_{s\in\cS}\sum_{i\in\cQ} c_{is}\cdot z'_{is},
    \end{equation*}
\end{restatable}
\begin{proof}
    Since $\mathrm{small}(T, \beta)$ is feasible, there exists a solution $z'$.
    Further, we obtain a (fractional) solution $z''$ to $\mathrm{small}(T/\beta, 1)$ 
    by dividing all variables of $z'$ by $\beta$.
    Using the approach by Shmoys and Tardos~\cite{ShmoysT93},
    we round $z''$ to an integral solution.
    For the sake of completeness, we provide the proof below.

    Assume without loss of generality that $\cS = \{1,2,\dotsc,|\cS|\}$ is ordered such that $v_s \ge v_{s-1}$ for all $s=2,3,\dotsc,|\cS|$.
    We construct an auxiliary bipartite graph $G$ (see \cref{fig:smallRounding} for an illustration).
    \begin{figure}
        \centering
        \includegraphics[scale=0.4]{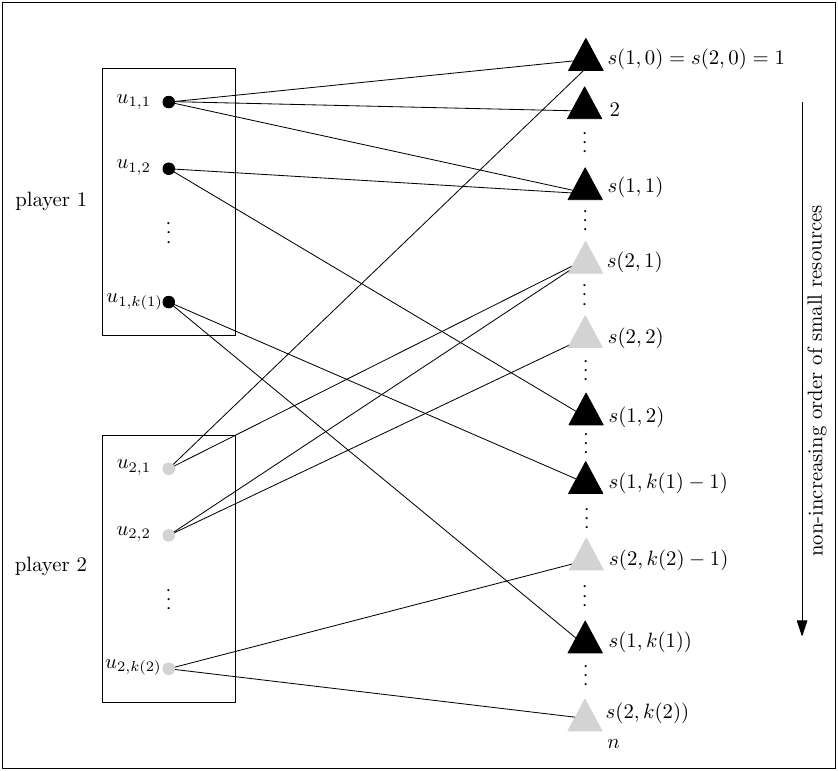}
        \caption{Bipartite graph G where on left side there are $k(i)$ copies of each player $i$ for an instance of two players with edges $(u_{i,\ell}, s)$ of weight $c_{is}$ for each $s = s(i,\ell-1),s(i,\ell-1)+1,\dotsc,s(i,\ell)$.} \label{fig:smallRounding}
    \end{figure}
    The elements
    on the right side of~$G$ are elements of
    $\cS$.
    On the left side, there are~$k(i) := \lfloor \sum_{s\in\cS} z''_{is} \rfloor$
    many vertices $u_{i, 1},\dotsc,u_{i,k(i)}$ for each
    $i\in\cQ$.
    In the fractional solution~$z''$, a player $Q$
    can get several small resources.
    Let~$k(i)$ denote their (rounded) number.
    Introducing~$u_{i,1},\dotsc,u_{i,k(i)}$ as
    copies of player $i\in Q$ allows us to argue about matchings
    (where every vertex is only involved in one assignment).
    Suppose we add one edge between every two vertices
    of the different sides of the bipartite graph.
    It is straight-forward that there exists a (fractional)
    left-perfect matching by distributing the resources
    assigned in $z''$ among
    the copies of each player.
    Due to the rounding in $k(i)$, this matching gives a slightly lower value to each player, but stays within the bounds we are aiming for.
    
    To round to a good integral matching, however, we require a specific definition of the
    fractional left-perfect matching.
    Essentially, we need a monotone assignment where
    the first copy $u_{i,1}$ of player $i\in Q$ has the highest value resources (the first in the order above) and the last player the lowest value resources.
    Then~$G$ only contains the edges that 
    are actually used in this assignment.

    This requires some careful definitions.
    For each $i\in \cQ$, we set
    $s(i,0) = 1$. This describes the first resource that
    can be fractionally assigned to player $i$.
    For $\ell = 1,2,\dotsc,k(i)$,
    we choose $s(i,\ell)$ as the element in $\cS$ that satisfies
    \begin{equation*}
        z''_{i,1} + \cdots + z''_{i,s(i,\ell)-1} < \ell
        \text{ and }
        z''_{i,1} + \cdots + z''_{i,s(i,\ell)} \ge \ell .
    \end{equation*}
    Note that $s(i,\ell)$ exists, because the sum of
    all $z''_{i,j}$ is at least $k(i) \ge \ell$.
    We only assign resources $s(i,\ell-1),\dotsc,s(i,\ell)$ to copy~$u_{i,\ell}$.
    Intuitively, the resources $1,2,\dotsc,s(i,\ell-1)-1$ should be exclusively used for the copies $u_{i,1},\dotsc,u_{i,\ell-1}$. Simply because the sum of fractions associated with $i$ (according to~$z''$) is not enough to cover the copies and $u_{i,\ell}$ should not receive any of them in order to maintain the monotonicity goals.
    Similarly, as the sum of fractions
    of resources $1,2,\dotsc,s(i,\ell)$ belonging to $i$
    exceeds $\ell$, they are enough to cover all players in
    $u_{i,1},\dotsc,u_{i,\ell}$.
    Thus, it is not necessary to give any less valuable resources to 
    $u_{i,\ell}$.

    Consequently, for all $i\in\cQ$ and $\ell=1,2,\dotsc,k(i)$, we
    introduce an edge $(u_{i,\ell}, s)$ of weight~$c_{is}$ for
    each $s = s(i,\ell-1),s(i,\ell-1)+1,\dotsc,s(i,\ell)$.
    We now formally show that there is a left-perfect
    fractional matching
    of weight at most $\sum_{s\in\cS}\sum_{i\in\cQ}c_{is} \cdot z''_{is}$.
    Towards this,
    consider some $i\in \cQ$ and $\ell \in \{1,2,\dotsc,k(i)\}$.
    We select edge $(u_{i,\ell}, s(i, \ell-1))$ to an
    extend of
    \begin{equation*}
         z''_{i,1} + \cdots + z''_{i,s(i, \ell-1)} - (\ell - 1) \in [0, z''_{i,s(i, \ell-1)}].
    \end{equation*}
    For $s = s(i, \ell-1)+1,\dotsc,s(i, \ell)-1$, we pick edge $(u_{i,\ell}, s)$ to an extend of $z''_{i,s}$.
    Finally, we choose edge $(u_{i,\ell}, s(i, \ell))$ to an
    extend of
    \begin{equation*}
        \ell - (z''_{i,1} + \cdots + z''_{i,s(i, \ell)-1}) \in [0, z''_{i,s(i, \ell)}] .
    \end{equation*}
    This fractional selection of edges satisfies the following properties.
    \smallskip
    \begin{enumerate}
        \item For each $i\in\cQ$ and $\ell\in\{1,2,\dotsc,k(i)\}$, the total fractional amount of selected 
        edges that are incident to~$u_{i,\ell}$ 
        is exactly $1$.
        \item For each $i\in\cQ$ and $s\in\cS$,
        the total fractional amount of selected edges that are between~$u_{i,1},\dotsc,u_{i,k(i)}$ and $s$
        is at most $z''_{i,s}$.
    \end{enumerate}
    Property~2 implies that the total fractional amount that edges incident to $s$ (over all $i\in\cQ$) are selected is at most $1$.
    Hence, the total weight is
    at most the cost of $z''$.
    For a bipartite graph, the set of all fractional
    matchings is precisely the convex hull of integral
    matchings, see e.g.~\cite[Chapter~18]{schrijver2003combinatorial}. Furthermore, $z''$ must lie on a facet spanned by only left-perfect
    matchings.
    Thus, there must also
    exist an integral left-perfect matching $M$ of weight at most
    \begin{equation*}
        \sum_{i\in\cQ} \sum_{s\in\cS} z''_{i, s} c_{i, s}
        \le \sum_{i\in\cQ} \sum_{s\in\cS} z'_{i, s} c_{i, s} .
    \end{equation*}
    We can find such a matching using standard algorithms
    for minimum weight bipartite matching.
    We interpret $M$ as an integral assignment $Z$
    where each $i\in \cQ$ receives all $s\in \cS$,
    for which there is an edge $(u_{i,\ell}, s)$ in $M$
    for some $\ell\in\{1,2,\dotsc,k(i)\}$.
    Finally, we analyze how much value each $i\in\cQ$
    receives in this assignment.
    Notice that all resources $s\in\cS$
    that are connected to $u_{i,\ell}$ have a value
    of at least $v_{s(i,\ell)}$.
    Therefore, $i$ receives a total value of at least
    \begin{equation*}
        v_{s(i,1)} + \cdots + v_{s(i,k(i))} .
    \end{equation*}
    On the other hand, 
    since $z''_{i,1} + \cdots + z''_{i,s(i, \ell ) - 1} + z''_{i,s(i, \ell)} < \ell + z''_{i,s(i, \ell)} \le \ell + 1$ for each $\ell$,
    the total fractional amount of resources that $i$ receives in $z''$
    with value at least $v_{s(i, \ell)}$ is less than $\ell + 1$.
    As a result, the total value that $i$ received in $z''$
    is at most
    \begin{equation*}
        v_{s(i,0)} + \cdots + v_{s(i,k(i))} .
    \end{equation*}
    This is at least $T/\beta$, because of $z''\in\mathrm{small}(T/\beta,1)$.
    As a consequence, in $Z$ player $i$ receives a value of at least
    $T/\beta - v_{s(i,0)} \ge T/\beta - \max_{s\in\cS} v_s$.
\end{proof}

\subsection{Approximation Factor}
Concluding the previous subsections, this section provides an $\alpha$-approximation for the Budgeted Santa Claus Problem, where $\alpha = \Order(\log n)$. 
\SCapprox*
\begin{proof}
   Let $\beta = O(\log n)$ be the value from Property~3 of \Cref{lem:cost-preservation}. We define $\alpha = 2\beta$
   and run our binary search over value $T$,
   which is then used in our definition of big and small resources.
   Then we solve the LP relaxation \LP{T}. 
   Let $(x,z)$ be the resulting solution. 
   If the cost of the solution is more than $C$, we fail (and
   increase the value of $T$). Assume now that $(x,z)$ has
   a cost of at most $C$.
   In \Cref{lem:cost-preservation}, using the dependent rounding procedure from \Cref{thm:assign}, we show that the fractional assignments $x$ of big resources to players can be rounded to an integral assignment $X$ and the assignment of small resources $z$ can be changed to $z'$ (which is still fractional), such that with high probability each small resource up to
   $\beta$ times and the cost is still below $C$. 
   \Cref{lem:rounding-small-items} proves that $z'$
   can be rounded to an integral assignment $Z$ such that
   the cost does not increase and each player $i$ that does
   not receive a big resource gets small resources of total value at least
   \begin{align*}
    \frac{T}{\beta} - \max_{s\in\cS} v_s &\geq \frac{2T}{\alpha} - \frac{T}{\alpha} = \frac{T}{\alpha}. \qedhere
    \end{align*}
\end{proof}

\section{Conclusion}
Based on the finding in this paper there are several interesting questions arising for future research.
For the Budgeted Santa Claus Problem, a naturally arising question is whether the approximation factor of
$O(\log n)$ can be improved. Notably, the special case of restricted assignment (without a budget constraint)
admits a constant approximation due to Feige~\cite{Feige08}.
We are not aware of any hardness
results indicating that such a result cannot hold for our problem. As an intermediate question,
one could look at the bi-criteria approximation that approximates both the minimum player value and the cost by a constant. This would still generalize the aforementioned algorithm for restricted assignment.

Another intriguing question is whether a dependent rounding scheme exists for rounding matroid bases that simultaneously guarantees cost preservation and Chernoff-type concentration 
like \textsc{SwapRounding}~\cite{ChekuriVZ10} does. 
It seems likely that the techniques from this paper
transfer at least to a limited
class of matroids, namely strongly base-orderable
matroids, because
these have very strong decomposition properties for the
symmetric difference of two bases.
It might, however, require other ideas to generalize to
arbitrary matroids.

\bibliographystyle{plainurl}
\bibliography{paper}

\begin{thebibliography}{10}

\bibitem{AnnamalaiKS17}
Chidambaram Annamalai, Christos Kalaitzis, and Ola Svensson.
\newblock Combinatorial algorithm for restricted max-min fair allocation.
\newblock {\em {ACM} Trans. Algorithms}, 13(3):37:1--37:28, 2017.
\newblock \href {https://doi.org/10.1145/3070694} {\path{doi:10.1145/3070694}}.

\bibitem{AsadpourFS08}
Arash Asadpour, Uriel Feige, and Amin Saberi.
\newblock Santa claus meets hypergraph matchings.
\newblock In {\em Approximation, Randomization and Combinatorial Optimization.
  Algorithms and Techniques, 11th International Workshop ({APPROX} 2008) and
  12th International Workshop ({RANDOM} 2008)}, volume 5171 of {\em Lecture
  Notes in Computer Science}, pages 10--20. Springer, 2008.
\newblock \href {https://doi.org/10.1007/978-3-540-85363-3\_2}
  {\path{doi:10.1007/978-3-540-85363-3\_2}}.

\bibitem{AsadpourGMGS17}
Arash Asadpour, Michel~X. Goemans, Aleksander Madry, Shayan~Oveis Gharan, and
  Amin Saberi.
\newblock An \emph{O}(log \emph{n}/log log \emph{n})-approximation algorithm
  for the asymmetric traveling salesman problem.
\newblock {\em Oper. Res.}, 65(4):1043--1061, 2017.
\newblock URL: \url{https://doi.org/10.1287/opre.2017.1603}, \href
  {https://doi.org/10.1287/OPRE.2017.1603} {\path{doi:10.1287/OPRE.2017.1603}}.

\bibitem{BamasLMRS24}
{\'{E}}tienne Bamas, Alexander Lindermayr, Nicole Megow, Lars Rohwedder, and
  Jens Schl{\"{o}}ter.
\newblock Santa claus meets makespan and matroids: Algorithms and reductions.
\newblock In {\em 35th Annual {ACM-SIAM} Symposium on Discrete Algorithms,
  ({SODA} 2024)}, pages 2829--2860. {SIAM}, 2024.
\newblock \href {https://doi.org/10.1137/1.9781611977912.100}
  {\path{doi:10.1137/1.9781611977912.100}}.

\bibitem{BansalN16}
Nikhil Bansal and Viswanath Nagarajan.
\newblock Approximation-friendly discrepancy rounding.
\newblock In {\em 18th Integer Programming and Combinatorial Optimization
  Conference ({IPCO} 2016)}, volume 9682 of {\em Lecture Notes in Computer
  Science}, pages 375--386. Springer, 2016.
\newblock \href {https://doi.org/10.1007/978-3-319-33461-5\_31}
  {\path{doi:10.1007/978-3-319-33461-5\_31}}.

\bibitem{BansalS06}
Nikhil Bansal and Maxim Sviridenko.
\newblock The santa claus problem.
\newblock In {\em 38th Annual {ACM} Symposium on Theory of Computing, ({STOC}
  2006)}, pages 31--40. {ACM}, 2006.
\newblock \href {https://doi.org/10.1145/1132516.1132522}
  {\path{doi:10.1145/1132516.1132522}}.

\bibitem{ChekuriVZ10}
Chandra Chekuri, Jan Vondr{\'{a}}k, and Rico Zenklusen.
\newblock Dependent randomized rounding via exchange properties of
  combinatorial structures.
\newblock In {\em 51st Annual {IEEE} Symposium on Foundations of Computer
  Science ({FOCS} 2010)}, pages 575--584. {IEEE} Computer Society, 2010.
\newblock \href {https://doi.org/10.1109/FOCS.2010.60}
  {\path{doi:10.1109/FOCS.2010.60}}.

\bibitem{ChekuriVZ11}
Chandra Chekuri, Jan Vondr{\'{a}}k, and Rico Zenklusen.
\newblock Multi-budgeted matchings and matroid intersection via dependent
  rounding.
\newblock In {\em 22nd Annual {ACM-SIAM} Symposium on Discrete Algorithms
  ({SODA} 2011)}, pages 1080--1097. {SIAM}, 2011.
\newblock \href {https://doi.org/10.1137/1.9781611973082.82}
  {\path{doi:10.1137/1.9781611973082.82}}.

\bibitem{DaviesRZ20}
Sami Davies, Thomas Rothvoss, and Yihao Zhang.
\newblock A tale of santa claus, hypergraphs and matroids.
\newblock In {\em 31st Annual {ACM-SIAM} Symposium on Discrete Algorithms,
  ({SODA} 2020)}, pages 2748--2757. {SIAM}, 2020.
\newblock \href {https://doi.org/10.1137/1.9781611975994.167}
  {\path{doi:10.1137/1.9781611975994.167}}.

\bibitem{Feige08}
Uriel Feige.
\newblock On allocations that maximize fairness.
\newblock In {\em 19th Annual {ACM-SIAM} Symposium on Discrete Algorithms
  ({SODA} 2008)}, pages 287--293. {SIAM}, 2008.
\newblock URL: \url{http://dl.acm.org/citation.cfm?id=1347082.1347114}.

\bibitem{GandhiKPS06}
Rajiv Gandhi, Samir Khuller, Srinivasan Parthasarathy, and Aravind Srinivasan.
\newblock Dependent rounding and its applications to approximation algorithms.
\newblock {\em J. {ACM}}, 53(3):324--360, 2006.
\newblock \href {https://doi.org/10.1145/1147954.1147956}
  {\path{doi:10.1145/1147954.1147956}}.

\bibitem{GrandoniZ10}
Fabrizio Grandoni and Rico Zenklusen.
\newblock Approximation schemes for multi-budgeted independence systems.
\newblock In {\em 18th Annual European Symposium ({ESA} 2010)}, volume 6346 of
  {\em Lecture Notes in Computer Science}, pages 536--548. Springer, 2010.
\newblock \href {https://doi.org/10.1007/978-3-642-15775-2\_46}
  {\path{doi:10.1007/978-3-642-15775-2\_46}}.

\bibitem{JansenKV20}
Klaus Jansen, Kim{-}Manuel Klein, and Jos{\'{e}} Verschae.
\newblock Closing the gap for makespan scheduling via sparsification
  techniques.
\newblock {\em Math. Oper. Res.}, 45(4):1371--1392, 2020.
\newblock URL: \url{https://doi.org/10.1287/moor.2019.1036}, \href
  {https://doi.org/10.1287/MOOR.2019.1036} {\path{doi:10.1287/MOOR.2019.1036}}.

\bibitem{LovettM15}
Shachar Lovett and Raghu Meka.
\newblock Constructive discrepancy minimization by walking on the edges.
\newblock {\em {SIAM} J. Comput.}, 44(5):1573--1582, 2015.
\newblock \href {https://doi.org/10.1137/130929400}
  {\path{doi:10.1137/130929400}}.

\bibitem{SahaS18}
Barna Saha and Aravind Srinivasan.
\newblock A new approximation technique for resource-allocation problems.
\newblock {\em Random Struct. Algorithms}, 52(4):680--715, 2018.
\newblock URL: \url{https://doi.org/10.1002/rsa.20756}, \href
  {https://doi.org/10.1002/RSA.20756} {\path{doi:10.1002/RSA.20756}}.

\bibitem{schrijver2003combinatorial}
Alexander Schrijver et~al.
\newblock {\em Combinatorial optimization: polyhedra and efficiency},
  volume~24.
\newblock Springer, 2003.

\bibitem{ShmoysT93}
David~B. Shmoys and {\'{E}}va Tardos.
\newblock An approximation algorithm for the generalized assignment problem.
\newblock {\em Math. Program.}, 62:461--474, 1993.
\newblock \href {https://doi.org/10.1007/BF01585178}
  {\path{doi:10.1007/BF01585178}}.

\end{thebibliography}

\appendix
\section{Appendix}

\subsection{Santa Claus Problem where all resources need to be assigned}
\label{subsec:appendix-santa-claus-with-all-resources-assigned}
In the following, we show that the Budgeted Santa Claus Problem with the requirement that all items need to be assigned is not harder than the variant we study. 
For each resource $j$, adjust the cost of assigning $j$ to any player $i$ to $c_{ij}' = c_{ij} - \min_{i \in \mathcal{P}} c_{ij}$. This reduces the total budget by the sum of these minimum values across all resources, $C' = C - \sum_{j \in \mathcal{R}} \min_{i \in \mathcal{P}} c_{ij}$. Then we solve the reduced problem without enforcing that all the resources have to be assigned. If some resources remain unassigned, we allocate them to the players with zero cost (i.e., players $i \in \cP$ with $c_{ij}'  = 0$). This ensures that all resources are assigned without exceeding the original budget, as the reduced budget $C'$ already accounted for these assignments. This reduction maintains the approximation ratio of the solution, as values remain the same and costs are not approximated.

\end{document}